\documentclass[11pt]{article}
\usepackage{times}
\usepackage[paperwidth=199.8mm,
paperheight=297mm,centering,hmargin=20mm,vmargin=20mm]{geometry}
\usepackage{authblk} % package for author list
\usepackage[bottom]{footmisc} % footnote

\usepackage[titletoc,title]{appendix} % appendix
\usepackage{changepage}

\usepackage{hyperref}
\hypersetup{colorlinks=true,citecolor=myblue,linkcolor=myblue,
filecolor=myblue,urlcolor=myblue,breaklinks=true}
\usepackage{url}
\usepackage[dvipsnames]{xcolor}
\usepackage{color}
\usepackage{framed}

\definecolor{shadecolor}{rgb}{0.9,0.9,0.9}
\definecolor{mylightgray}{RGB}{100,100,100}

\definecolor{myblue}{RGB}{0, 68, 116}
\definecolor{mycyan}{RGB}{0, 97, 91}
\definecolor{mygreen}{RGB}{2, 102, 1}
\definecolor{myorange}{RGB}{240, 102, 0}
\definecolor{myred}{RGB}{172, 23, 0}
\definecolor{mymagenta}{RGB}{140,16,73}

\usepackage{mathtools}
\usepackage{amsmath}
\usepackage{amssymb}
\usepackage[shortlabels]{enumitem}
\usepackage{graphicx,epsfig,latexsym,verbatim}
\usepackage{dsfont}
\usepackage{mathrsfs}

\usepackage{subcaption}
\usepackage{caption}
\usepackage{float}
\usepackage{tikz}
\usetikzlibrary{arrows.meta,positioning,calc,fit,matrix,backgrounds}
\usepackage{relsize}
\usepackage{pgfplots}
\pgfplotsset{compat=1.18}

\usepackage{amsthm}
\usepackage{tcolorbox}
\tcbuselibrary{skins}
\tcbuselibrary{breakable}

\definecolor{myblue}{RGB}{0,68,116}
\definecolor{mycyan}{RGB}{0,97,91}
\definecolor{mygreen}{RGB}{2,102,1}
\definecolor{myorange}{RGB}{240,102,0}
\definecolor{myred}{RGB}{172,23,0}

\newtheorem{theorem}{Theorem}

\newtheorem{corollary}[theorem]{Corollary}

\newtheorem{boxedtheorem}[theorem]{Theorem}
\newtheorem{boxedlemma}[theorem]{Lemma}
\newtheorem{boxedprop}[theorem]{Proposition}
\newtheorem{boxedcorollary}[theorem]{Corollary}
\newtheorem{boxeddefinition}[theorem]{Definition}
\newtheorem{boxedassumption}[theorem]{Assumption}
\newtheorem{boxedremark}[theorem]{Remark}

\newtheorem{boxedexample}[theorem]{Example}

\newtheorem{remark}[theorem]{Remark}

\definecolor{theoremframe}{RGB}{31,78,121}
\definecolor{theoremback}{RGB}{247,250,253}
\definecolor{lemmaframe}{RGB}{35,105,86}
\definecolor{lemmaback}{RGB}{247,252,250}
\definecolor{propframe}{RGB}{102,73,139}
\definecolor{propback}{RGB}{251,249,253}
\definecolor{corollaryframe}{RGB}{165,96,18}
\definecolor{corollaryback}{RGB}{254,251,246}
\definecolor{definitionframe}{RGB}{112,112,112}
\definecolor{definitionback}{RGB}{248,248,248}

\tcbset{
  supplement result/.style={
    enhanced jigsaw,
    breakable,
    lines before break=5,
    sharp corners,
    boxrule=0.55pt,
    left=6pt,
    right=6pt,
    oversize=0pt,
    top=6pt,
    bottom=6pt,
    before skip=9pt plus 2pt minus 1pt,
    after skip=9pt plus 2pt minus 1pt
  }
}
\tcolorboxenvironment{boxedtheorem}{
  supplement result,
  colframe=myblue,
  colback=myblue!5!white
}
\tcolorboxenvironment{boxedlemma}{
  supplement result,
  colframe=definitionframe,
  colback=definitionback
}
\tcolorboxenvironment{boxedprop}{
  supplement result,
  colframe=definitionframe,
  colback=definitionback
}
\tcolorboxenvironment{boxedcorollary}{
  supplement result,
  colframe=definitionframe,
  colback=definitionback
}
\tcolorboxenvironment{boxeddefinition}{
  supplement result,
  colframe=definitionframe,
  colback=definitionback
}
\tcolorboxenvironment{boxedassumption}{
  supplement result,
  colframe=definitionframe,
  colback=definitionback
}
\tcolorboxenvironment{boxedremark}{
  supplement result,
  colframe=myorange!70!black,
  colback=myorange!6!white
}
\tcolorboxenvironment{boxmessage}{
  supplement result,
  colframe=myorange!70!black,
  colback=myorange!6!white
}
\newtcolorbox{boxnote}{
  supplement result,
  colframe=myorange!70!black,
  colback=myorange!6!white
}
\tcolorboxenvironment{boxedexample}{
  supplement result,
  colframe=definitionframe,
  colback=definitionback
}

\newenvironment{boxtheorem}[1][]
  {\begin{boxedtheorem}[#1]}
  {\end{boxedtheorem}}
\newenvironment{boxlemma}[1][]
  {\begin{boxedlemma}[#1]}
  {\end{boxedlemma}}
\newenvironment{boxproposition}[1][]
  {\begin{boxedprop}[#1]}
  {\end{boxedprop}}

\usepackage{colortbl}
\usepackage{pifont} % http://ctan.org/pkg/pifont
\usepackage{booktabs}
\usepackage{makecell} %To keep spacing of text in tables
\usepackage{diagbox}
\usepackage{multirow}

\newcommand{\nc}{\newcommand}
\nc{\rnc}{\renewcommand}
\nc{\<}{\langle}
\rnc{\>}{\rangle}
\nc{\bra}[1]{\langle#1|}
\nc{\ket}[1]{|#1\rangle}
\nc{\ketbra}[2]{|#1\rangle\!\langle#2|}
\nc{\braket}[2]{\langle#1|#2\rangle}
\nc{\braandket}[3]{\langle #1|#2|#3\rangle}
\nc{\proj}[1]{| #1\rangle\!\langle #1 |}
\nc{\avg}[1]{\langle#1\rangle}
\nc{\vket}[1]{\lvert#1\rangle\!\rangle}
\nc{\vbra}[1]{\langle\!\langle#1\rvert}

\nc{\rank}{\operatorname{Rank}}
\nc{\id}{{\operatorname{id}}}
\nc{\iid}{{\operatorname{iid}}}
\nc{\supp}{{\operatorname{supp}}}
\nc{\smfrac}[2]{\mbox{$\frac{#1}{#2}$}}
\nc{\tr}{\operatorname{Tr}}
\nc{\ox}{\otimes}
\nc{\floor}[1]{\lfloor #1 \rfloor}
\nc{\trans}{\mathsf T}
\nc{\img}{\mathbf{i}}

\makeatletter

\newcommand{\solidgamearrowpicture@}[3]{%
    \tikz[x=1ex,y=1ex,line cap=round,line join=round]
      \draw[line width=#3,-{Triangle[length=0.64ex,width=0.68ex]}] (0,0) -- (#1,#2);%
}
\newcommand{\opengamearrowpicture@}[3]{%
    \tikz[x=1ex,y=1ex,line cap=round,line join=round]
      \draw[line width=#3,-{Triangle[open,length=0.64ex,width=0.68ex]}] (0,0) -- (#1,#2);%
}
\newsavebox{\solidgameuparrow@display}
\newsavebox{\solidgameuparrow@text}
\newsavebox{\solidgameuparrow@script}
\newsavebox{\solidgameuparrow@scriptscript}
\newsavebox{\solidgamedownarrow@display}
\newsavebox{\solidgamedownarrow@text}
\newsavebox{\solidgamedownarrow@script}
\newsavebox{\solidgamedownarrow@scriptscript}
\newsavebox{\opengameuparrow@display}
\newsavebox{\opengameuparrow@text}
\newsavebox{\opengameuparrow@script}
\newsavebox{\opengameuparrow@scriptscript}
\newsavebox{\opengamedownarrow@display}
\newsavebox{\opengamedownarrow@text}
\newsavebox{\opengamedownarrow@script}
\newsavebox{\opengamedownarrow@scriptscript}
\AtBeginDocument{%
  \sbox{\solidgameuparrow@display}{\solidgamearrowpicture@{0}{1.62}{0.58pt}}%
  \sbox{\solidgameuparrow@text}{\solidgamearrowpicture@{0}{1.50}{0.52pt}}%
  \sbox{\solidgameuparrow@script}{\solidgamearrowpicture@{0}{1.30}{0.44pt}}%
  \sbox{\solidgameuparrow@scriptscript}{\solidgamearrowpicture@{0}{1.12}{0.38pt}}%
  \sbox{\solidgamedownarrow@display}{\solidgamearrowpicture@{0}{-1.62}{0.58pt}}%
  \sbox{\solidgamedownarrow@text}{\solidgamearrowpicture@{0}{-1.50}{0.52pt}}%
  \sbox{\solidgamedownarrow@script}{\solidgamearrowpicture@{0}{-1.30}{0.44pt}}%
  \sbox{\solidgamedownarrow@scriptscript}{\solidgamearrowpicture@{0}{-1.12}{0.38pt}}%
  \sbox{\opengameuparrow@display}{\opengamearrowpicture@{0}{1.62}{0.58pt}}%
  \sbox{\opengameuparrow@text}{\opengamearrowpicture@{0}{1.50}{0.52pt}}%
  \sbox{\opengameuparrow@script}{\opengamearrowpicture@{0}{1.30}{0.44pt}}%
  \sbox{\opengameuparrow@scriptscript}{\opengamearrowpicture@{0}{1.12}{0.38pt}}%
  \sbox{\opengamedownarrow@display}{\opengamearrowpicture@{0}{-1.62}{0.58pt}}%
  \sbox{\opengamedownarrow@text}{\opengamearrowpicture@{0}{-1.50}{0.52pt}}%
  \sbox{\opengamedownarrow@script}{\opengamearrowpicture@{0}{-1.30}{0.44pt}}%
  \sbox{\opengamedownarrow@scriptscript}{\opengamearrowpicture@{0}{-1.12}{0.38pt}}%
}
\newcommand{\gamearrowbox@}[1]{%
  \mathrel{\vcenter{\hbox{%
    \usebox{#1}%
  }}}%
}
\DeclareRobustCommand{\soliduparrow}{%
  \mathchoice
    {\gamearrowbox@{\solidgameuparrow@display}}%
    {\gamearrowbox@{\solidgameuparrow@text}}%
    {\gamearrowbox@{\solidgameuparrow@script}}%
    {\gamearrowbox@{\solidgameuparrow@scriptscript}}%
}
\DeclareRobustCommand{\soliddownarrow}{%
  \mathchoice
    {\gamearrowbox@{\solidgamedownarrow@display}}%
    {\gamearrowbox@{\solidgamedownarrow@text}}%
    {\gamearrowbox@{\solidgamedownarrow@script}}%
    {\gamearrowbox@{\solidgamedownarrow@scriptscript}}%
}
\DeclareRobustCommand{\emptyuparrow}{%
  \mathchoice
    {\gamearrowbox@{\opengameuparrow@display}}%
    {\gamearrowbox@{\opengameuparrow@text}}%
    {\gamearrowbox@{\opengameuparrow@script}}%
    {\gamearrowbox@{\opengameuparrow@scriptscript}}%
}
\DeclareRobustCommand{\emptydownarrow}{%
  \mathchoice
    {\gamearrowbox@{\opengamedownarrow@display}}%
    {\gamearrowbox@{\opengamedownarrow@text}}%
    {\gamearrowbox@{\opengamedownarrow@script}}%
    {\gamearrowbox@{\opengamedownarrow@scriptscript}}%
}

\makeatother

\nc{\cA}{{\cal A}}
\nc{\cB}{{\cal B}}
\nc{\cC}{{\cal C}}
\nc{\cD}{{\cal D}}
\nc{\cE}{{\cal E}}
\nc{\cF}{{\cal F}}
\nc{\cG}{{\cal G}}
\nc{\cH}{{\cal H}}
\nc{\cI}{{\cal I}}
\nc{\cJ}{{\cal J}}
\nc{\cK}{{\cal K}}
\nc{\cL}{{\cal L}}
\nc{\cM}{{\cal M}}
\nc{\cN}{{\cal N}}
\nc{\cO}{{\cal O}}
\nc{\cP}{{\cal P}}
\nc{\cQ}{{\cal Q}}
\nc{\cR}{{\cal R}}
\nc{\cS}{{\cal S}}
\nc{\cT}{{\cal T}}
\nc{\cV}{{\cal V}}
\nc{\cU}{{\cal U}}
\nc{\cX}{{\cal X}}
\nc{\cY}{{\cal Y}}
\nc{\cZ}{{\cal Z}}
\nc{\cW}{{\cal W}}

\nc{\bp}{\boldsymbol{p}}
\nc{\bq}{\boldsymbol{q}}
\nc{\brho}{\boldsymbol{\rho}}
\nc{\bsigma}{\boldsymbol{\sigma}}
\nc{\bomega}{\boldsymbol{\omega}}
\nc{\bmu}{\boldsymbol{\mu}}
\nc{\bT}{\boldsymbol{T}}
\nc{\bS}{\boldsymbol{S}}

\nc{\RR}{{{\mathbb R}}}
\nc{\CC}{{{\mathbb C}}}
\nc{\FF}{{{\mathbb F}}}
\nc{\NN}{{{\mathbb N}}}
\nc{\ZZ}{{{\mathbb Z}}}
\nc{\QQ}{{{\mathbb Q}}}
\nc{\UU}{{{\mathbb U}}}
\nc{\EE}{{{\mathbb E}}}

\nc{\bH}{{\mathfrak{H}}}

\nc{\sK}{{{\mathscr{K}}}}
\nc{\sS}{{{\mathscr{S}}}}
\nc{\sT}{{{\mathscr{T}}}}
\nc{\sA}{{{\mathscr{A}}}}
\nc{\sB}{{{\mathscr{B}}}}
\nc{\sC}{{{\mathscr{C}}}}
\nc{\sE}{{{\mathscr{E}}}}
\nc{\sL}{{{\mathscr{L}}}}
\nc{\sG}{{{\mathscr{G}}}}
\nc{\sF}{{{\mathscr{F}}}}
\nc{\sI}{{{\mathscr{I}}}}
\nc{\sN}{{{\mathscr{N}}}}
\nc{\sM}{{{\mathscr{M}}}}

\nc{\Choi}{Choi-Jamio\l{}kowski }
\nc{\reg}{\infty}
\nc{\mx}{\text{\rm mx}}
\nc{\amo}{\text{\rm amo}}
\nc{\Renyi}{R\'{e}nyi }

\nc{\conv}{\operatorname{conv}}
\nc{\cvxset}{\mathscr{C}}

\nc{\RM}{{{\mathscr{R}}}}

\nc{\END}{\operatorname{End}}
\nc{\PERM}{\mathfrak{\sigma}}

\nc{\Cone}{\text{\rm Cone}}
\nc{\sep}{{\SEP}}

\nc{\DD}{{{\mathbb D}}}
\nc{\BS}{{\scriptscriptstyle \rm {BS}}}
\nc{\Sand}{{\scriptscriptstyle  \rm S}}
\nc{\Petz}{{\scriptscriptstyle  \rm P}}
\nc{\Hypo}{{\scriptscriptstyle  \rm H}}
\nc{\Meas}{{\scriptscriptstyle \rm M}}
\nc{\Proj}{{{\scriptscriptstyle \rm P}}}
\nc{\KL}{{{\scriptscriptstyle \rm KL}}}
\nc{\IS}{{{\scriptscriptstyle \rm IS}}}

\nc{\suchthat}{\text{\rm s.t.}}

\nc{\pl}{{\scalebox{0.7}{+}}}
\nc{\HERM}{\mathscr{H}}
\nc{\PSD}{\HERM_{\pl}}
\nc{\PD}{\HERM_{\pl\pl}}
\nc{\density}{\mathscr{D}}
\nc{\subdensity}{\mathscr{D}_\bullet}

\nc{\polarPSD}[1]{{#1}_{\pl}^{\circ}}
\nc{\polarPSDre}[1]{{#1}_{\pl}^{\star}}
\nc{\polarPD}[1]{{#1}_{\pl\pl}^{\circ}}

\nc{\PPT}{\text{\rm PPT}}
\nc{\Rains}{\text{\rm Rains}}
\nc{\WD}{\text{\rm WD}}
\nc{\SEP}{\text{\rm SEP}}
\nc{\PSEP}{\text{\rm PSEP}}
\nc{\CPTP}{\text{\rm CPTP}}
\nc{\POVM}{\text{\rm POVM}}
\nc{\PVM}{\text{\rm PVM}}
\nc{\CP}{\text{\rm CP}}
\nc{\adv}{\text{\rm adv}}
\nc{\spec}{\text{\rm spec}}
\nc{\poly}{\text{\rm poly}}
\nc{\End}{\operatorname{End}}
\nc{\Par}{\operatorname{Par}}
\nc{\RNG}{\operatorname{RNG}}
\nc{\STAB}{\text{\rm STAB}}
\nc{\epi}{\boldsymbol{\operatorname{epi}}}
\nc{\op}{\boldsymbol{\operatorname{op}}}

\makeatletter
\newcommand*\rel@kern[1]{\kern#1\dimexpr\macc@kerna}
\newcommand*\widebar[1]{%
  \begingroup
  \def\mathaccent##1##2{%
    \rel@kern{0.8}%
    \overline{\rel@kern{-0.8}\macc@nucleus\rel@kern{0.2}}%
    \rel@kern{-0.2}%
  }%
  \macc@depth\@ne
  \let\math@bgroup\@empty \let\math@egroup\macc@set@skewchar
  \mathsurround\z@ \frozen@everymath{\mathgroup\macc@group\relax}%
  \macc@set@skewchar\relax
  \let\mathaccentV\macc@nested@a
  \macc@nested@a\relax111{#1}%
  \endgroup
}
\makeatother

\usepackage[T1]{fontenc}
\usepackage[utf8]{inputenc}
\usepackage{microtype}
\usepackage{placeins}

\usetikzlibrary{arrows.meta,positioning,calc}

\colorlet{darkblue}{myblue}
\hypersetup{
  colorlinks=true,
  linkcolor=darkblue,
  citecolor=darkblue,
  filecolor=darkblue,
  urlcolor=darkblue
}

\newcommand{\diag}{\operatorname{diag}}
\newcommand{\one}{{I}}
\newcommand{\Ef}{E_{\mathrm F}}
\newcommand{\Ec}{E_{\mathrm C}}

\newcommand{\W}{W}
\newcommand{\FME}{F}
\newcommand{\Neg}{N}
\newcommand{\ME}{\mathrm{ME}}
\newcommand{\PhiBell}{\Phi_0}

\definecolor{knownframe}{HTML}{667085}
\definecolor{knownback}{HTML}{F2F4F7}
\definecolor{newframe}{HTML}{175CD3}
\definecolor{newback}{HTML}{EFF8FF}
\definecolor{derivedframe}{HTML}{027A78}
\definecolor{derivedback}{HTML}{ECFDF3}
\definecolor{consequenceframe}{HTML}{B54708}
\definecolor{consequenceback}{HTML}{FFFAEB}

\begin{document}

\title{\Large \textbf{Entanglement cost of quantum depolarization}}
\author[1]{Kun Fang}
\affil[1]{\small School of Data Science, The Chinese University of Hong Kong, Shenzhen,\protect\\
  Guangdong 518172, China}
\date{\today}

\maketitle

\begin{abstract}
Entanglement cost is the asymptotic rate of Bell pairs required to prepare a quantum state. Its regularized definition has made exact evaluation difficult, even for isotropic states (i.e., depolarized maximally entangled states). In this work, we determine the entanglement cost of every qubit isotropic state and, more generally, every two-qubit Bell-diagonal state. Our proof uses a family of suitably tuned qubit semigroups to transform a general log-Sobolev entropy bound into supporting lines of Wootters' function. A recent exact tensorization theorem of Dong et al.\ [arXiv:2606.17729] extends these bounds to arbitrarily many copies, yielding a general lower bound on entanglement cost. Matching this bound with the entanglement of formation also determines the exact cost for a broader class of two-qubit states. For qudit isotropic states, we derive another general lower bound from the full depolarizing $2\to 3$ norm. This bound substantially improves on the PPT-relative entropy of entanglement and nearly matches the entanglement of formation, leaving a gap of at most $2\%$ of $\log d$ that vanishes as $d$ grows. Finally, the exact qubit results and qudit bounds extend to the entanglement cost of preparing quantum depolarizing channels under both parallel and sequential strategies.
\end{abstract}

{
\small
\tableofcontents
}

\section{Introduction}

Entanglement cost quantifies the asymptotic rate of Bell pairs required to prepare a bipartite state by local operations and classical communication. It is characterized by the regularized entanglement of formation, so an exact one-copy formula does not generally determine the operational cost~\cite{HaydenHorodeckiTerhal2001}. This obstruction already appears for isotropic states, obtained by subjecting one half of a maximally entangled state to depolarizing noise. Their entanglement of formation is known in every dimension, and Wootters' formula covers all two-qubit states~\cite{Wootters1998,TerhalVollbrecht2000}. Nevertheless, even the entanglement cost of qubit isotropic states has remained unresolved. The central difficulty is to exclude an asymptotic saving from preparing many copies collectively.

We resolve this problem for qubits. For every two-qubit state $\rho$, we prove the lower bound $\Ec(\rho)\geq\W(\FME(\rho))$ in terms of its normalized entangled fraction (Theorem~\ref{thm:main-state-lower}). This bound matches Wootters' formula whenever the normalized negativity equals the concurrence, yielding $\Ec(\rho)=\Ef(\rho)$ for this precisely characterized class (Proposition~\ref{prop:main-exact-cost}). In particular, collective preparation offers no asymptotic saving for any two-qubit Bell-diagonal state, including every qubit isotropic state (Theorem~\ref{thm:main-bell-diagonal}). The proof converts a general log-Sobolev entropy estimate into supporting lines of Wootters' function through a suitably tuned family of qubit semigroups. The exact tensorization theorem in Ref.~\cite[Theorem~3.14]{DongEtAl2026} then promotes these one-pair bounds to arbitrary tensor powers.

For dimensions $d\geq3$, the exact cost remains open, but we obtain a narrow rigorous interval. The full depolarizing $2\to3$ norm yields a general cubic-norm lower bound for every $d\times d$ state (Theorem~\ref{thm:refined-cubic}). For isotropic states, pairing this bound with the known entanglement of formation gives a certified interval that substantially improves the PPT-relative-entropy benchmark~\cite{Rains1999} and nearly matches the entanglement of formation, leaving a gap of at most $2\%$ of $\log d$ that vanishes as $d$ grows.

These state results also determine the entanglement resources needed to simulate quantum depolarizing channels. A known reduction identifies both the parallel and adaptive sequential simulation costs of a depolarizing channel with the entanglement cost of its isotropic normalized Choi state~\cite[Theorem~1 and Section~IV]{Wilde2018}. Consequently, we obtain the exact entanglement cost of every qubit depolarizing channel under both simulation criteria, together with a pretty-tight estimation for qudit depolarizing channels.

\section{Preliminaries}
We collect the state and semigroup conventions used in the proofs.

\paragraph{Notation.}
All Hilbert spaces are finite dimensional. We write $\sL(A)$ and $\density(A)$ for the linear operators and density operators on $A$, and $\one_A$ for the identity. A pure-state projector is abbreviated as $\psi:=\proj\psi$. The Schatten norms are $\|X\|_s:=(\tr|X|^s)^{1/s}$ for $1\leq s<\infty$, where $|X|:=\sqrt{X^\dagger X}$.
Unless stated otherwise, $\log$ is base two and $\ln$ is natural. We use $S(\rho):=-\tr(\rho\log\rho)$ for the von Neumann entropy and $h_2(p):=-p\log p-(1-p)\log(1-p)$ for the binary entropy. Relative entropy is measured in nats: $D(\rho\Vert\sigma):=\tr[\rho(\ln\rho-\ln\sigma)]$ for $\sigma>0$.

\paragraph{Maximally entangled states and isotropic states.}
For $d\geq2$, the standard maximally entangled vector and projector are
\begin{align}
\ket{\Phi^{(d)}}
&:=\frac1{\sqrt d}\sum_{a=0}^{d-1}\ket a\ket a,
&
\Phi^{(d)}
&:=\proj{\Phi^{(d)}}.
\label{eq:main-standard-maximally-entangled-state}
\end{align}
For qubits, we abbreviate $\ket{\PhiBell}:=\ket{\Phi^{(2)}}$. Let $\sigma_0:=\one$ and $(\sigma_1,\sigma_2,\sigma_3):=(X,Y,Z)$. The Bell basis is fixed as
\begin{align}
\ket{\Phi_\mu}
&:=(\one\otimes\sigma_\mu)\ket{\PhiBell},
\qquad \mu\in\{0,1,2,3\}.
\label{eq:main-bell-basis}
\end{align}

For a state $\rho$ with equal local dimensions, define its normalized entangled fraction by
\begin{align}
\FME(\rho)
&:=2\max_{\Phi\in\ME}\tr(\rho\Phi)-1.
\label{eq:main-bell-certificate}
\end{align}
Here $\ME$ is the set of maximally entangled rank-one projectors in the relevant dimension.

Depolarizing $\Phi^{(d)}$ with noise parameter $0\leq p\leq1$ gives the isotropic states
\begin{align}
\rho_{p,d}
&:=(1-p)\Phi^{(d)}+p\frac{\one_{d^2}}{d^2}.
\label{eq:main-isotropic-choi-family}
\end{align}
For qubits, we abbreviate $\rho_p:=\rho_{p,2}$.

\paragraph{Entanglement cost of states.}
The entanglement cost $\Ec(\omega)$ is the minimum asymptotic Bell-pair rate required to prepare copies of a bipartite state $\omega$ by local operations and classical communication, with vanishing error. The entanglement of formation is
\begin{align}
\Ef(\omega)
&:=\inf_{\omega=\sum_xp_x\proj{\psi_x}}
\sum_xp_xS((\psi_x)_A),
\end{align}
where the infimum is over pure-state decompositions. Regularization gives~\cite{HaydenHorodeckiTerhal2001}
\begin{align}
\Ec(\omega)
&=\lim_{n\to\infty}\frac{1}{n}
\Ef\!\left(\omega^{\otimes n}\right).
\label{eq:main-state-costs}
\end{align}
For a two-qubit state $\omega$, let $p_1\geq p_2\geq p_3\geq p_4$ be the eigenvalues of $\sqrt{\sqrt{\omega}\,\widetilde{\omega}\sqrt{\omega}}$, where $\widetilde{\omega}:=(Y\otimes Y)\omega^*(Y\otimes Y)$. Its concurrence is $C(\omega):=\max\{0,p_1-p_2-p_3-p_4\}$. Wootters' formula gives~\cite{Wootters1998}
\begin{align}
\Ef(\omega)&=\W(C(\omega)),&
\W(c)
&:=
h_2\!\left(\dfrac{1+\sqrt{1-c^2}}{2}\right), \quad 0\leq c\leq1.
\end{align}
We extend $\W$ by setting $\W(c)=0$ for $c<0$. The extended function is nondecreasing and convex, with $\W'(0^+)=0$; these properties underlie the supporting-line argument in Section~\ref{sec:main-qubit}.

\paragraph{Quantum Markov semigroups.}
We use the conventions of Ref.~\cite[Sections~2.1 and~2.4]{DongEtAl2026} for weighted norms and entropy. For a full-rank density operator $\sigma$, put $\Gamma_\sigma^s(X):=\sigma^{s/2}X\sigma^{s/2}$ for $s\in\mathbb R$. The weighted Schatten norms and Kubo--Martin--Schwinger (KMS) inner product are
\begin{align}
\lVert X\rVert_{\sigma,p}
&:=\left\lVert\Gamma_{\sigma}^{1/p}(X)\right\rVert_p,\qquad 1\leq p<\infty,\\
\langle X,Y\rangle_\sigma
&:=\tr\!\left(X^\dagger\sigma^{1/2}Y\sigma^{1/2}\right).
\label{eq:main-weighted-norms}
\end{align}

A quantum Markov semigroup (QMS) is a continuous family $(\Psi_t)_{t\geq0}$ of completely positive unital maps satisfying $\Psi_0=\id$ and $\Psi_{s+t}=\Psi_s\circ\Psi_t$. We use the positive-generator convention $\Psi_t=e^{-t\cL}$. A state $\sigma$ is invariant if $\tr[\sigma\Psi_t(X)]=\tr(\sigma X)$ for all $t\geq0$ and $X$. The QMS is primitive if it has a unique invariant state that is full rank, and $\sigma$-reversible if $\cL$ is self-adjoint with respect to $\langle\cdot,\cdot\rangle_\sigma$~\cite[Definitions~2.7 and~2.10]{DongEtAl2026}.

The QMS has positive off-diagonal scaling (PODS) in a fixed eigenbasis of $\sigma$ if every $\Psi_t$ preserves the diagonal algebra and multiplies each off-diagonal matrix unit by a nonnegative scalar~\cite[Definition~2.11]{DongEtAl2026}.

\paragraph{Log-Sobolev inequalities.}
For a primitive $\sigma$-reversible QMS, the $2$-Dirichlet form is
\begin{align}
\cE_{2,\cL}(X)
&:=\langle X,\cL(X)\rangle_\sigma.
\label{eq:main-weighted-dirichlet}
\end{align}
For $X>0$, put $M:=\sigma^{1/4}X\sigma^{1/4}$. Its weighted $2$-entropy is
\begin{align}
\operatorname{Ent}_{2,\sigma}(X)
&:=
\tr\!\left[M^2(\ln M^2-\ln\sigma)\right]
-\tr(M^2)\ln\tr(M^2).
\label{eq:main-weighted-entropy}
\end{align}
The best constant in the log-Sobolev inequality
\begin{align}
\alpha\,\operatorname{Ent}_{2,\sigma}(X)
&\leq\cE_{2,\cL}(X),
\qquad X>0,
\label{eq:main-weighted-lsi}
\end{align}
is denoted by $\alpha_2(\cL)$.

\section{Entanglement cost for qubit states}
\label{sec:main-qubit}

We derive an entangled-fraction lower bound for two-qubit states and characterize a class for which it is exact. This class includes all Bell-diagonal states and hence all qubit isotropic states.

\subsection{A general lower bound for qubit states}

We begin with a lower bound that holds for every two-qubit state.

\begin{boxtheorem}[Entangled-fraction bound]
\label{thm:main-state-lower}
Let $|A|=|B|=2$. For any $\rho \in \density(AB)$, 
\begin{align}
\Ec(\rho)
&\geq\W\!\left(\FME(\rho)\right).
\label{eq:main-state-lower}
\end{align}
\end{boxtheorem}

The proof begins with a general entropy bound derived from a log-Sobolev inequality. Its right-hand side contains a reference-state term and a Dirichlet form. We construct a qubit semigroup for which these terms combine into a supporting line of $\W$. The exact tensorization theorem~\cite[Theorem~3.14]{DongEtAl2026} extends this bound to arbitrary pure states across many pairs. Averaging over pure-state decompositions and taking the regularized limit then proves Theorem~\ref{thm:main-state-lower}.

\begin{boxlemma}[General entropy bound]
\label{lem:main-general-entropy-bound}
Let $\Sigma$ be a full-rank density operator on $\CC^d$, and let $\cL$ generate a primitive $\Sigma$-reversible QMS with $\alpha_2(\cL)\geq a>0$. For any $M\in\sL(\CC^d)$ with $\tr(M^\dagger M)=1$, put $X:=\Sigma^{-1/4}M\Sigma^{-1/4}$. Then
\begin{align}
(\ln2)S(M^\dagger M)
&\geq-\frac12\tr\!\left[(M^\dagger M+MM^\dagger)\ln\Sigma\right]
-\frac1a\cE_{2,\cL}(X).
\label{eq:main-general-entropy-bound}
\end{align}
\end{boxlemma}

\begin{proof}
This bound holds for arbitrary complex $M$ and does not require a product reference state or a qubit generator. It follows from the Dirichlet-form comparison in Ref.~\cite[Lemma~23]{BeigiDattaRouze2020}.

Define the two positive semidefinite operators
\begin{align}
A&:=\Gamma_\Sigma^{-1/2}(|M|),
&B&:=\Gamma_\Sigma^{-1/2}(|M^\dagger|).
\end{align}
These are $I_{2,2}(X)$ and $I_{2,2}(X^\dagger)$ in the notation of Ref.~\cite{BeigiDattaRouze2020}. Since $|M|^2=M^\dagger M$ and $|M^\dagger|^2=MM^\dagger$ are density operators, Eq.~\eqref{eq:main-weighted-entropy} gives
\begin{align}
\operatorname{Ent}_{2,\Sigma}(A)&=D(M^\dagger M\Vert\Sigma),
&\operatorname{Ent}_{2,\Sigma}(B)&=D(MM^\dagger\Vert\Sigma).
\end{align}
Applying the log-Sobolev inequality to $A$ and $B$ and then using Ref.~\cite[Lemma~23]{BeigiDattaRouze2020} yields
\begin{align}
a\left[D(M^\dagger M\Vert\Sigma)+D(MM^\dagger\Vert\Sigma)\right]
&\leq\cE_{2,\cL}(A)+\cE_{2,\cL}(B)
\leq\cE_{2,\cL}(X)+\cE_{2,\cL}(X^\dagger)
=2\cE_{2,\cL}(X).
\label{eq:main-general-dirichlet-comparison}
\end{align}
To justify the last equality, note that $\cL$ preserves adjoints. Cyclicity of the trace then gives $\cE_{2,\cL}(X^\dagger)=\overline{\cE_{2,\cL}(X)}$, while $\Sigma$-reversibility gives self-adjointness in the weighted inner product. Hence
\begin{align}
\overline{\cE_{2,\cL}(X)}
&=\langle\cL(X),X\rangle_\Sigma
=\langle X,\cL(X)\rangle_\Sigma
=\cE_{2,\cL}(X).
\end{align}
Thus the Dirichlet form is real even when $X$ is not Hermitian. The log-Sobolev inequality extends from positive definite to positive semidefinite inputs by continuity, so the argument also covers singular $M$. Finally, expanding the two relative entropies gives
\begin{align}
\frac12D(M^\dagger M\Vert\Sigma)+\frac12D(MM^\dagger\Vert\Sigma)
&=-(\ln2)S(M^\dagger M)
-\frac12\tr\!\left[(M^\dagger M+MM^\dagger)\ln\Sigma\right],
\label{eq:main-general-entropy-expansion}
\end{align}
because $M^\dagger M$ and $MM^\dagger$ have the same eigenvalues. Combining this identity with Eq.~\eqref{eq:main-general-dirichlet-comparison} proves Eq.~\eqref{eq:main-general-entropy-bound}.
\end{proof}

For a contact point $0<c\leq1$, write the supporting line of Wootters' function as
\begin{align}
L_c(f)&:=\W(c)+\W'(c)(f-c).
\label{eq:main-tuned-supporting-line}
\end{align}
We choose the reference state and generator so that the right-hand side of Eq.~\eqref{eq:main-general-entropy-bound} has this affine dependence on normalized Bell overlap.

\begin{boxlemma}[Supporting-line entropy bound]
\label{lem:main-supporting-line-entropy}
Let $0<c\leq1$ and let $n\geq1$ be an integer. For a normalized pure state $\psi$ on $A^nB^n$, with $|A_i|=|B_i|=2$ for every $i$, set $f_i:=2\tr(\psi_{A_iB_i}\PhiBell)-1$. Then
\begin{align}
S(\psi_{B^n})
&\geq\sum_{i=1}^nL_c(f_i).
\label{eq:main-pure-supporting-line-bound}
\end{align}
\end{boxlemma}

\begin{proof}
We apply Lemma~\ref{lem:main-general-entropy-bound} with a reference state and generator chosen to produce the supporting line $L_c$. Take $\sigma_c:=\diag(\mu_c,\nu_c)$ with
\begin{align}
\mu_c
&:=\frac{1-\sqrt{1-c^2}}{2}, \qquad
\nu_c
:=\frac{1+\sqrt{1-c^2}}{2}.
\label{eq:main-tuned-parameters}
\end{align}
For this reference state, define the qubit generator by
\begin{align}
\cL_c
\begin{pmatrix}
x&z\\ w&y
\end{pmatrix}
&:=
\begin{pmatrix}
\nu_c(x-y)&\eta_c z\\
\eta_c w&\mu_c(y-x)
\end{pmatrix}.
\label{eq:main-tuned-generator}
\end{align}
For $0<c<1$, set
\begin{align}
a_c&:=\frac{\sqrt{1-c^2}}{\ln(\nu_c/\mu_c)},
&
\eta_c&:=\frac12+\sqrt{\mu_c\nu_c}=\frac{1+c}{2}.
\label{eq:main-tuned-rates}
\end{align}
At $c=1$, take $a_1=1/2$ and $\eta_1=1$ by continuity. Appendix~\ref{app:tuned-qubit-lsi} shows that $\Psi_{c,t}:=e^{-t\cL_c}$ is a primitive $\sigma_c$-reversible PODS semigroup with $\alpha_2(\cL_c)=a_c$. To treat $n$ pairs, use the product reference state and the sum of local generators:
\begin{align}
\Sigma_c&:=\sigma_c^{\otimes n},
&\cL_c^{(n)}&:=\sum_{i=1}^n\cL_{c,i},
&\cL_{c,i}&:=\id^{\otimes(i-1)}\otimes\cL_c\otimes\id^{\otimes(n-i)}.
\end{align}
The exact tensorization theorem~\cite[Theorem~3.14]{DongEtAl2026} therefore applies and preserves the log-Sobolev constant for every $n$:
\begin{align}
\alpha_2(\cL_c^{(n)})&=\alpha_2(\cL_c)=a_c.
\label{eq:main-tensorized-lsi}
\end{align}

To estimate the entanglement of $\psi$, write $\ket{\psi}=\sum_{\mathbf a,\mathbf b}M_{\mathbf a,\mathbf b}\ket{\mathbf a}_{A^n}\ket{\mathbf b}_{B^n}$ in the computational bases and set $X:=\Sigma_c^{-1/4}M\Sigma_c^{-1/4}$. Normalization gives $\tr(M^\dagger M)=1$, and the Schmidt decomposition gives $S(\psi_{B^n})=S(M^\dagger M)$. Applying Lemma~\ref{lem:main-general-entropy-bound} therefore yields
\begin{align}
(\ln2)S(\psi_{B^n})
&\geq-\frac12\tr\!\left[(M^\dagger M+MM^\dagger)\ln\Sigma_c\right]
-\frac1{a_c}\cE_{2,\cL_c^{(n)}}(X).
\label{eq:main-common-tangent-entropy-bound}
\end{align}

It remains to evaluate the two terms on the right. Both $\cL_c^{(n)}$ and $\ln\Sigma_c$ are sums of local terms, so the calculation reduces to the pair marginals even when $\psi$ is entangled across pairs. For $\rho_i:=\psi_{A_iB_i}$, set $p_{ab}^{(i)}:=\bra{ab}\rho_i\ket{ab}$ and $z_i:=\bra{00}\rho_i\ket{11}$. Appendix~\ref{app:tuned-qubit-matching} gives the resulting expressions:
\begin{align}
\cE_{2,\cL_c^{(n)}}(X)
&=\sum_{i=1}^n\left[\nu_c p_{00}^{(i)}+\mu_c p_{11}^{(i)}
+\eta_c\bigl(p_{01}^{(i)}+p_{10}^{(i)}\bigr)
-2\sqrt{\mu_c\nu_c}\operatorname{Re}z_i\right],
\label{eq:main-common-tangent-energy}\\
-\frac12\tr\!\left[(M^\dagger M+MM^\dagger)\ln\Sigma_c\right]
&=\sum_{i=1}^n\left[-p_{00}^{(i)}\ln\mu_c-p_{11}^{(i)}\ln\nu_c
-\frac{p_{01}^{(i)}+p_{10}^{(i)}}2\ln(\mu_c\nu_c)\right].
\label{eq:main-common-tangent-reference}
\end{align}
Since $\sum_{a,b}p_{ab}^{(i)}=1$, the Bell overlap satisfies $f_i=2\operatorname{Re}z_i-p_{01}^{(i)}-p_{10}^{(i)}$. The rates in Eq.~\eqref{eq:main-tuned-rates} make the two terms above combine into an affine function of this overlap. Substituting these expressions into Eq.~\eqref{eq:main-common-tangent-entropy-bound} yields
\begin{align}
S(\psi_{B^n})
&\geq\sum_{i=1}^n\left[\W(c)-c\W'(c)+\W'(c)f_i\right]
=\sum_{i=1}^nL_c(f_i).
\label{eq:main-common-tangent-entropy}
\end{align}
\end{proof}

\begin{proof}[Proof of Theorem~\ref{thm:main-state-lower}]
Fix $0<c\leq1$. Choose a maximally entangled projector attaining $\FME(\rho)$ and rotate it to $\PhiBell$. Apply the same local change of basis to every copy of $\rho$. For an arbitrary pure-state decomposition $\rho^{\otimes n}=\sum_jr_j\psi_j$, Lemma~\ref{lem:main-supporting-line-entropy} gives
\begin{align}
\sum_jr_jS((\psi_j)_{B^n})
&\geq n[\W(c)-c\W'(c)]
+\W'(c)\sum_{i=1}^n\sum_jr_j
\bigl(2\tr[(\psi_j)_{A_iB_i}\PhiBell]-1\bigr)\nonumber\\
&=n\bigl[\W(c)-c\W'(c)+\W'(c)\FME(\rho)\bigr].
\label{eq:main-homogeneous-convex-roof}
\end{align}
The equality follows from linearity and the fact that each pair marginal $(\rho^{\otimes n})_{A_iB_i}$ equals $\rho$. The ensemble vectors need not factor across copies. Taking the infimum over all decompositions gives the same lower bound for $\Ef(\rho^{\otimes n})$.

The extended function $\W$ is the supremum of its supporting lines:
\begin{align}
\sup_{0<c\leq1}
\bigl[\W(c)-c\W'(c)+x\W'(c)\bigr]
&=\W(x),\qquad -1\leq x\leq1.
\label{eq:main-tangent-envelope}
\end{align}
For $x\leq0$, the supremum is approached as $c\to0^+$, so the singular invariant state at $c=0$ is never needed. Optimizing over $c$ in Eq.~\eqref{eq:main-homogeneous-convex-roof} therefore gives
\begin{align}
\Ef(\rho^{\otimes n})
&\geq n\W\!\left(\FME(\rho)\right).
\label{eq:main-state-lower-finite-power}
\end{align}
Dividing by $n$ and taking $n\to\infty$ in Eq.~\eqref{eq:main-state-costs} proves the claim.
\end{proof}

\bigskip
\begin{remark}
Ref.~\cite[Theorem~3.14]{DongEtAl2026} provides the exact tensorized constant $a_c$ required by Eq.~\eqref{eq:main-common-tangent-entropy}. For $0<c<1$, the tuned generator $\cL_c$ has the same diagonal action as the simple generalized depolarizing generator with invariant state $\sigma_c$, but its off-diagonal rate is $\eta_c=(1+c)/2<1$ instead of $1$. It is also not Hilbert--Schmidt self-adjoint. Thus neither the generalized-depolarizing results~\cite[Theorems~21 and~24]{BeigiDattaRouze2020} nor King's theorem~\cite[Theorem~1]{King2014} covers this family. The general estimate in Ref.~\cite[Corollary~26]{BeigiDattaRouze2020} gives only $\eta_c a_c<a_c$, which is insufficient for the supporting-line matching.
\end{remark}

\subsection{Tightness of the entangled-fraction bound}

We now identify states for which the lower bound is tight. Define the normalized negativity $\Neg(\rho):=\lVert\rho^\Gamma\rVert_1-1$, where $\Gamma$ denotes partial transpose, and consider the class
\begin{align}
\cS_{\mathrm{ex}}
&:=\{\rho\in\density(AB): \dim A=\dim B=2,\ \Neg(\rho)=C(\rho)\}.
\label{eq:main-exact-class}
\end{align}
Every separable two-qubit state belongs to $\cS_{\mathrm{ex}}$ because both $\Neg$ and $C$ vanish. The following lemma characterizes which entangled states belong to this class.

\begin{boxlemma}[Characterization of the exact class]
\label{lem:main-state-tightness}
Every two-qubit state $\rho$ satisfies
$\FME(\rho)
\leq\Neg(\rho)\leq C(\rho)$.
If $\rho$ is entangled, the following conditions are equivalent:
\begin{enumerate}[(i),leftmargin=2.8em,itemsep=1pt,topsep=3pt]
\item $\rho\in\cS_{\mathrm{ex}}$;
\item $\FME(\rho)=\Neg(\rho)$;
\item $\FME(\rho)=C(\rho)$;
\item the one-dimensional negative eigenspace of $\rho^\Gamma$ is spanned by a maximally entangled vector.
\end{enumerate}

\end{boxlemma}

\begin{proof}
Each maximally entangled projector $\Phi$ determines another such projector $\Xi_\Phi$ through the partial-transpose identity below. The correspondence $\Phi\mapsto\Xi_\Phi$ is bijective, giving
\begin{align}
\Phi^\Gamma
&=\frac{\one_4}{2}-\Xi_\Phi,
&
\FME(\rho)
&=-2\min_{\Xi\in\ME}\tr(\rho^\Gamma\Xi).
\label{eq:main-bell-partial-transpose}
\end{align}
If $\rho$ is separable, then $\FME(\rho)\leq0=\Neg(\rho)=C(\rho)$, so the inequalities hold. If $\rho$ is entangled, its partial transpose has a unique negative eigenvalue $\lambda_-(\rho^\Gamma)$. Rayleigh--Ritz and the sharp negativity--concurrence comparison~\cite{VerstraeteEtAl2001} give
\begin{align}
\FME(\rho)
&\leq-2\lambda_-(\rho^\Gamma)
=\Neg(\rho)
\leq C(\rho).
\label{eq:main-tightness-chain-proof}
\end{align}
Rayleigh--Ritz shows that condition~(ii) is equivalent to condition~(iv). The equality characterization for $\Neg(\rho)\leq C(\rho)$ in Ref.~\cite{VerstraeteEtAl2001}, together with Eq.~\eqref{eq:main-exact-class}, shows that condition~(i) is equivalent to condition~(iv). Finally, condition~(iii) holds exactly when both inequalities in Eq.~\eqref{eq:main-tightness-chain-proof} are equalities. These observations establish the equivalence of conditions~(i)--(iv).
\end{proof}

On $\cS_{\mathrm{ex}}$, this characterization makes the entangled-fraction lower bound coincide with the entanglement of formation, yielding the exact cost.

\begin{boxproposition}[Exact cost when negativity equals concurrence]
\label{prop:main-exact-cost}
For every $\rho\in\cS_{\mathrm{ex}}$,
\begin{align}
\Ec(\rho)&=\Ef(\rho)=\W(C(\rho)).
\label{eq:main-exact-class-cost}
\end{align}
\end{boxproposition}

\begin{proof}
For an entangled state in $\cS_{\mathrm{ex}}$, Theorem~\ref{thm:main-state-lower}, Lemma~\ref{lem:main-state-tightness}, and Wootters' formula give
\begin{align}
\W(C(\rho))
&=\W(\FME(\rho))
\leq\Ec(\rho)\leq\Ef(\rho)=\W(C(\rho)).
\label{eq:main-exact-class-sandwich}
\end{align}
The upper bound follows by preparing the copies independently. For separable states, all quantities vanish, completing the proof.
\end{proof}

In particular, this gives the exact cost of qubit isotropic states and, more generally, all Bell-diagonal states.

\begin{boxtheorem}
\label{thm:main-bell-diagonal}
For every $0\leq p\leq1$, the qubit isotropic state satisfies
\begin{align}
\Ec(\rho_p)
&=\Ef(\rho_p)
=\W\!\left(1-\frac{3p}{2}\right).
\label{eq:main-isotropic-cost}
\end{align}
More generally, let $\rho=\sum_{z=0}^3q_z\Phi_z$ be a Bell-diagonal state and set $q_{\max}:=\max_zq_z$. Then
\begin{align}
\Ec(\rho)&=\Ef(\rho)=\W(2q_{\max}-1).
\label{eq:main-bell-diagonal-cost}
\end{align}
\end{boxtheorem}

\begin{proof}
Direct evaluation gives
$
\FME(\rho)=2q_{\max}-1$ and
$\Neg(\rho)=C(\rho)$.
Hence $\rho\in\cS_{\mathrm{ex}}$, and Proposition~\ref{prop:main-exact-cost} gives Eq.~\eqref{eq:main-bell-diagonal-cost}. Specializing to $\rho_p$, for which $\FME(\rho_p)=1-3p/2$, gives Eq.~\eqref{eq:main-isotropic-cost}.
\end{proof}

\section{Entanglement cost for qudit states}
\label{sec:main-qudit}

We derive an entangled-fraction lower bound for $d\times d$ states with $d\ge3$ and evaluate it for isotropic states. As in Section~\ref{sec:main-qubit}, the proof rests on affine entropy bounds valid for arbitrarily many pairs. Here we obtain these bounds from the standard depolarizing semigroup and its full $2\to3$ norm.

\subsection{A general lower bound for qudit states}
\label{sec:refined-cubic-main}

For $d\ge3$, write $\tau_d(X):=\tr(X)/d$ and $\|X\|_{s,\tau_d}:=\tau_d(|X|^s)^{1/s}$. The depolarizing map and its $2\to3$ norm are
\begin{align}
\Delta_r(X)&:=rX+(1-r)\tau_d(X)\one,&
K_d(r)&:=\sup_{X\ne0}\frac{\|\Delta_r(X)\|_{3,\tau_d}}{\|X\|_{2,\tau_d}},
\label{eq:refined-cubic-norm-def}
\end{align}
where $0\le r\le1$ and the supremum includes all complex matrices. For $0<r\le1$, define the affine functions and their envelope by
\begin{align}
\ell_{d,r}(f)&:=\log d+3(1-f)\log r-6\log K_d(r),
\label{eq:main-cubic-affine-function}\\
B_d(c)&:=\sup_{0<r\le1}\ell_{d,r}(c),\qquad -1\le c\le1.
\label{eq:refined-cubic-profile}
\end{align}
The term ``cubic-norm'' refers to the Schatten $3$-norm defining $K_d$. Since $K_d(1)=d^{1/6}$, the choice $r=1$ gives $\ell_{d,1}=0$, so $B_d$ is nonnegative.

\begin{boxtheorem}[Cubic-norm lower bound]
\label{thm:refined-cubic}
For every $d\times d$ state $\rho$ with $d\ge3$,
\begin{align}
\Ec(\rho)&\ge B_d(\FME(\rho)).
\label{eq:refined-cubic-cost}
\end{align}
\end{boxtheorem}

The proof follows the same route as the qubit argument. We first show that $\ell_{d,r}$ bounds pure-state entropy in terms of the local Bell overlaps, then average over preparation ensembles and optimize over $r$. The standard depolarizing generator is useful because its Dirichlet form equals the sum of the Bell-overlap deficits, while the maximally mixed reference state contributes a constant term. The required norm estimate is proved in Section~\ref{sec:main-cubic-proof}, and the resulting envelope is evaluated in Section~\ref{sec:main-cubic-profile}.

\begin{boxlemma}[Affine entropy bound]
\label{lem:main-cubic-affine-entropy}
Let $d\ge3$, $0<r\le1$, and $n\ge1$. For a normalized pure state $\psi$ on $A^nB^n$, with $|A_i|=|B_i|=d$ for every $i$, set $f_i:=2\tr(\psi_{A_iB_i}\Phi^{(d)})-1$. Then
\begin{align}
S(\psi_{B^n})&\ge\sum_{i=1}^n\ell_{d,r}(f_i).
\label{eq:main-cubic-affine-entropy}
\end{align}
\end{boxlemma}

\begin{proof}
Take the reference state $\sigma=\one_d/d$ and the generator $\cL_d(X):=X-\tau_d(X)\one$, with semigroup $e^{-t\cL_d}=\Delta_{e^{-t}}$. This generator projects onto the traceless matrices. Since the normalized identity $\one_d/\sqrt d$ is the coefficient matrix of $\ket{\Phi^{(d)}}$, this projection removes the component along $\ket{\Phi^{(d)}}$.

Put $D=d^n$ and $t=-\ln r$. Let $\mathcal E_i$ replace the $i$th matrix factor by its normalized trace, and use the sum of local generators:
\begin{align}
\mathcal E_i(X)&:=\frac{\one_d}{d}\otimes\tr_iX,&
\cL_d^{(n)}&:=\sum_{i=1}^n(\id-\mathcal E_i),&
e^{-t\cL_d^{(n)}}&=\Delta_r^{\otimes n}.
\end{align}
The identity in $\mathcal E_i$ is inserted in the original tensor position. Lemma~\ref{lem:refined-cubic-product} gives the norm estimate
\begin{align}
\|\Delta_r^{\otimes n}(X)\|_{3,\tau_D}
&\le K_d(r)^n\|X\|_{2,\tau_D}
\qquad\text{for every matrix }X.
\label{eq:main-cubic-tensor-power}
\end{align}
We convert this estimate into an entropy inequality by interpolating back to the $2$-norm and differentiating there. Since the maps $\id-\mathcal E_i$ are orthogonal projections with respect to $\langle X,Y\rangle_{\tau_D}:=\tau_D(X^\dagger Y)$, their sum $\cL_d^{(n)}$ is positive and self-adjoint, and $e^{-\mathrm iut\cL_d^{(n)}}$ is a $2$-norm isometry.

For fixed $X$, consider the bounded holomorphic family $F(z):=e^{-zt\cL_d^{(n)}}X$ on $0\le\Re z\le1$. The isometry and the all-operator estimate in Eq.~\eqref{eq:main-cubic-tensor-power} give, for every real $u$,
\begin{align*}
\|F(\mathrm iu)\|_{2,\tau_D}&=\|X\|_{2,\tau_D},\\
\|F(1+\mathrm iu)\|_{3,\tau_D}
&=\|\Delta_r^{\otimes n}(e^{-\mathrm iut\cL_d^{(n)}}X)\|_{3,\tau_D}
\le K_d(r)^n\|X\|_{2,\tau_D}.
\end{align*}
Apply the interpolation theorem~\cite[Theorem~2]{Beigi2013} with reference state $\one_D/D$ and boundary exponents $2$ and $3$. Its weighted norms are our normalized Schatten norms, so these boundary bounds yield
\begin{align}
\|e^{-\theta t\cL_d^{(n)}}X\|_{s(\theta),\tau_D}
&\le\|X\|_{2,\tau_D}^{1-\theta}
\bigl(K_d(r)^n\|X\|_{2,\tau_D}\bigr)^\theta
=K_d(r)^{n\theta}\|X\|_{2,\tau_D}.
\label{eq:main-cubic-interpolation}
\end{align}
Here $1/s(\theta)=(1-\theta)/2+\theta/3$ for $0\le\theta\le1$.

Normalize $\tau_D(X^\dagger X)=1$ and write $\operatorname{Ent}_{\tau_D}(X^\dagger X):=\tau_D(X^\dagger X\ln(X^\dagger X))$. The two sides of Eq.~\eqref{eq:main-cubic-interpolation} agree at $\theta=0$, so taking the right derivative of their natural logarithms yields
\begin{align}
\frac16\operatorname{Ent}_{\tau_D}(X^\dagger X)
-t\langle X,\cL_d^{(n)}X\rangle_{\tau_D}
&\le n\ln K_d(r).
\label{eq:main-cubic-entropy}
\end{align}
Indeed, $s'(0)=2/3$ and $\left.\partial_s\ln\|X\|_{s,\tau_D}\right|_{s=2}=\operatorname{Ent}_{\tau_D}(X^\dagger X)/4$, giving the entropy coefficient $1/6$. Differentiating the matrix argument contributes the Dirichlet term $-t\langle X,\cL_d^{(n)}X\rangle_{\tau_D}$. These identities hold for arbitrary complex $X$; singular matrices follow by an invertible approximation and continuity.

It remains to express the entropy and Dirichlet form in terms of $\psi$. Let $M$ be its coefficient matrix in the computational bases and put $X=\sqrt D\,M$, so $\tau_D(X^\dagger X)=\tr(M^\dagger M)=1$. The squared singular values of $M$ are the Schmidt probabilities, while tracing matching row and column indices contracts with the maximally entangled vector. Hence
\begin{align}
\operatorname{Ent}_{\tau_D}(X^\dagger X)
&=n\ln d-(\ln2)S(\psi_{B^n}),\\
\langle X,\cL_d^{(n)}X\rangle_{\tau_D}
&=\sum_{i=1}^n\left(1-\frac1d\|\tr_iM\|_2^2\right)
=\frac12\sum_{i=1}^n(1-f_i).
\label{eq:main-cubic-vectorization}
\end{align}
Here the displayed $2$-norm is unnormalized, and the identities hold even when $\psi$ correlates different pairs. Substituting into Eq.~\eqref{eq:main-cubic-entropy} and converting to base-two logarithms gives
\begin{align}
S(\psi_{B^n})
&\ge n\log d+3\log r\sum_{i=1}^n(1-f_i)-6n\log K_d(r)
=\sum_{i=1}^n\ell_{d,r}(f_i).
\end{align}
\end{proof}

\begin{proof}[Proof of Theorem~\ref{thm:refined-cubic}]
Choose a maximally entangled projector attaining $c:=\FME(\rho)$ and rotate it to $\Phi^{(d)}$, using the same local change of basis on every copy. Fix $0<r\le1$ and consider an arbitrary pure-state decomposition $\rho^{\otimes n}=\sum_jw_j\psi_j$. For $f_{ij}:=2\tr[(\psi_j)_{A_iB_i}\Phi^{(d)}]-1$, Lemma~\ref{lem:main-cubic-affine-entropy} gives
\begin{align}
\sum_jw_jS((\psi_j)_{B^n})
&\ge n\log d+3\log r\sum_{i=1}^n\left(1-\sum_jw_jf_{ij}\right)-6n\log K_d(r)
=n\ell_{d,r}(c).
\label{eq:main-cubic-ensemble}
\end{align}
The equality follows from $\sum_jw_jf_{ij}=c$ at every site, by linearity and the identical marginals of $\rho^{\otimes n}$. The same projector and parameter $r$ are fixed for all ensemble components, which may correlate different copies. Taking the infimum over decompositions, optimizing over $r$, and using Eq.~\eqref{eq:main-state-costs} gives
\begin{align}
\Ef(\rho^{\otimes n})&\ge nB_d(c),&
\Ec(\rho)&=\lim_{n\to\infty}\frac1n\Ef(\rho^{\otimes n})\ge B_d(c).
\label{eq:main-cubic-regularization}
\end{align}
\end{proof}

\subsection{The depolarizing cubic norm and its tensorization}
\label{sec:main-cubic-proof}

We now establish the norm estimate used in Lemma~\ref{lem:main-cubic-affine-entropy}. First we reduce the one-site norm to a scalar maximization; then we show that these norms multiply under tensor products. The cubic exponent makes both steps tractable through third moments and products of three matrix blocks. Write $c_d:=(d-2)/\sqrt{d-1}$, and use $\Delta_r^{(d)}$ when the dimension must be explicit.

\begin{boxlemma}[Exact cubic norm]
\label{lem:refined-cubic-norm}
For every $d\geq3$ and $0\leq r\leq1$,
\begin{align}
 K_d(r)^3=
 \max_{0\leq z\leq\sqrt{d-1}}
 \frac{1+3r^2z^2+c_dr^3z^3}{(1+z^2)^{3/2}}.
 \label{eq:refined-cubic-scalar-norm}
\end{align}
\end{boxlemma}
\begin{proof}
We first reduce the all-operator norm to positive inputs.
For a completely positive map $\Psi$, applying $\id_2\otimes\Psi$ to
\begin{align}
 \begin{pmatrix}|X^\dagger|&X\\X^\dagger&|X|\end{pmatrix}\geq0
\end{align}
preserves positivity and gives the factorization
$\Psi(X)=\Psi(|X^\dagger|)^{1/2}V\Psi(|X|)^{1/2}$ for a contraction
$V$ on the relevant supports. For $s\geq1$, Schatten H\"older gives
\begin{align}
 \|\Psi(X)\|_{s,\tau}
 \leq\|\Psi(|X^\dagger|)\|_{s,\tau}^{1/2}
      \|\Psi(|X|)\|_{s,\tau}^{1/2}.
 \label{eq:refined-cubic-cp-modulus}
\end{align}
Since both moduli have the same $2$-norm as $X$, a positive-input
$2\to s$ estimate implies the all-operator estimate with the same
constant. Singular cases follow by continuity. For $s=3$, this is the
tracial case of Ref.~\cite[Lemma~3.2]{DongEtAl2026}, with $n=1$,
$\sigma=\one/d$, $p=2$, and $q=3$.

By Eq.~\eqref{eq:refined-cubic-cp-modulus}, it suffices to consider $X\geq0$ with
$\tau_d(X^2)=1$. Write $m:=\tau_d(X)>0$, $v^2:=1-m^2$, and $z:=v/m$.
Positivity gives $m\geq d^{-1/2}$ and hence
$0\leq z\leq\sqrt{d-1}$. The sharp centered third-moment bound is
\begin{align}
 \tau_d\bigl((X-m\one)^3\bigr)\leq c_dv^3.
 \label{eq:refined-cubic-third-moment}
\end{align}
To prove the moment bound, maximize the third moment of $d$ real numbers with
zero mean and unit second moment. The Lagrange multiplier equations imply that
a stationary point has at most two distinct values.
If the positive value occurs $k$ times, the third moment is
$(d-2k)/\sqrt{k(d-k)}$. Its largest value is attained at $k=1$ and
equals $c_d$. Compactness of the constraint sphere gives the global
bound; $v=0$ is immediate.

Expanding the output cube gives
\begin{align}
 \tau_d\bigl(\Delta_r^{(d)}(X)^3\bigr)
 =m^3+3mr^2v^2+r^3\tau_d\bigl((X-m\one)^3\bigr).
\end{align}
For every admissible $z$, equality in Eq.~\eqref{eq:refined-cubic-third-moment}
is attained by the nonnegative spectrum
\begin{align}
 x_1=\frac{1+\sqrt{d-1}\,z}{\sqrt{1+z^2}},
 \qquad
 x_2=\cdots=x_d=
 \frac{1-z/\sqrt{d-1}}{\sqrt{1+z^2}}.
\end{align}
The moment bound gives the upper bound in Eq.~\eqref{eq:refined-cubic-scalar-norm}, and this spectrum attains it for every admissible $z$. The scalar maximization is therefore exact.
\end{proof}

To control collective preparation, we next establish multiplicativity of the one-site norm. We adapt the proof of Dong et al.~\cite[Theorem~3.7]{DongEtAl2026}, based on their norm-compression inequality \cite[Theorem~3.3]{DongEtAl2026}, while retaining the full $2\to3$ norm constants.

\begin{boxlemma}[Tensorization of the cubic norm]
\label{lem:refined-cubic-product}
For any dimensions $d_i\geq3$ and parameters $0\leq r_i\leq1$,
\begin{align}
 \left\|\bigotimes_{i=1}^n\Delta^{(d_i)}_{r_i}\right\|_{2\to3}
 =\prod_{i=1}^nK_{d_i}(r_i),
 \label{eq:refined-cubic-product}
\end{align}
where all norms use the normalized trace on their respective algebras.
\end{boxlemma}
\begin{proof}
Consider a completely positive map $\Psi':\sL(\CC^m)\to\sL(\CC^m)$ with all-operator $2\to3$ norm at most $C'$.
It suffices to prove the upper bound $K_d(r)C'$ for $\Delta_r^{(d)}\otimes\Psi'$. By the positive-input reduction, we may take $X\geq0$. Write its $m\times m$
blocks as $X_{ij}$, set $B_{ij}:=\Psi'(X_{ij})$, and let
$\widetilde B_{ij}:=\|B_{ij}\|_{3,\tau_m}$. Put
\begin{align}
 A:=(\Delta_r^{(d)}\otimes\Psi')(X),\qquad
 C:=\Delta_r^{(d)}(\widetilde B).
\end{align}
The scalar matrix $C$ is real symmetric and entrywise nonnegative.
The triangle inequality on diagonal blocks, and the exact off-diagonal
action, give $\|A_{ij}\|_{3,\tau_m}\leq C_{ij}$ for every $i,j$.
Expanding the cubic trace and applying H\"older to each product yields
\begin{align}
 \|A\|_{3,\tau_{dm}}^3
 &=\frac1d\sum_{i,j,k}\tau_m(A_{ij}A_{jk}A_{ki})\leq\frac1d\sum_{i,j,k}C_{ij}C_{jk}C_{ki}
 =\tau_d(C^3)\leq\tau_d(|C|^3).
 \label{eq:refined-cubic-compression}
\end{align}
Here $A\geq0$, but $C$ need not be positive semidefinite. It is
therefore essential that $K_d(r)$ is an all-operator norm. It follows that
\begin{align}
 \|A\|_{3,\tau_{dm}}
 &\leq K_d(r)\|\widetilde B\|_{2,\tau_d},\qquad
 \|\widetilde B\|_{2,\tau_d}^2
 =\frac1d\sum_{i,j}\|\Psi'(X_{ij})\|_{3,\tau_m}^2
 \leq (C')^2\|X\|_{2,\tau_{dm}}^2.
 \label{eq:refined-cubic-retained-constants}
\end{align}
This proves the two-factor upper bound $K_d(r)C'$, which iterates to the claimed product estimate. Product inputs attaining the one-site norms give the reverse inequality. Thus the exact norm constant is preserved even when the individual maps are not contractive.
\end{proof}

\subsection{Explicit evaluation of the cubic-norm bound}
\label{sec:main-cubic-profile}

It remains to optimize the affine entropy bounds over $r$. The scalar formula in Lemma~\ref{lem:refined-cubic-norm} makes this a one-variable problem. Define
\begin{align}
 u_d:=(d-1)^{1/3}+(d-1)^{-1/3},\qquad
 r_{0,d}:=u_d^{-1/2},\qquad
 c_{*,d}:=2(u_d-1)^{-2}-1.
 \label{eq:refined-cubic-critical-parameters}
\end{align}
Here $r_{0,d}$ is the largest parameter for which $K_d(r)=1$, and $c_{*,d}$ is the overlap at which the optimizer reaches this contractive endpoint. As $c$ increases, the optimal parameter moves from $r=1$, through an interior value, to $r=r_{0,d}$, giving the three branches below.
For $2/d-1<c<c_{*,d}$, set $z=\sqrt{(1-c)/(1+c)}$ and define
$r=r_d(c)\in(r_{0,d},1)$ as the unique root of
\begin{align}
 c_dzr^3+(2-z^2)r^2-1=0.
 \label{eq:refined-cubic-root}
\end{align}

\begin{boxproposition}[Explicit cubic-norm bound]
\label{prop:refined-cubic-explicit}
For every $d\geq3$,
\begin{align}
 B_d(c)=
 \begin{cases}
 0,
 &-1\leq c\leq 2/d-1,\\[6pt]
 \displaystyle
 \log\frac{2d(1+c)}{\bigl[1+c+(1-c)r^2\bigr]^2}
 +3(1-c)\log r,
 &2/d-1<c<c_{*,d},\\[14pt]
 \displaystyle
 \log d-\frac32(1-c)\log u_d,
 &c_{*,d}\leq c\leq1,
 \end{cases}
 \label{eq:refined-cubic-piecewise}
\end{align}
Here $r=r_d(c)$ is defined by Eq.~\eqref{eq:refined-cubic-root}.
\end{boxproposition}
\begin{proof}
We first locate the contractive range of the cubic norm, then optimize the entropy bound over the remaining parameters.
Denote the scalar objective in Eq.~\eqref{eq:refined-cubic-scalar-norm}
by $g_r(z)$. Its derivative has the sign of
\begin{align}
 z\bigl[(2r^2-1)+c_dr^3z-r^2z^2\bigr].
\end{align}
Using $c_d^2=(u_d-2)(u_d+1)^2$ shows that, at $r=r_{0,d}$, the global
maxima are $z=0$ and $z=\sqrt{u_d(u_d-2)}$, both with value one.
For each $z>0$, $g_r(z)$ is strictly increasing in $r$. The derivative
sign therefore gives $K_d(r)=1$ for $r\leq r_{0,d}$; for
$r_{0,d}<r\leq1$, its unique positive global maximizer is
\begin{align}
 z(r)=\frac{c_dr+\sqrt{c_d^2r^2+8-4/r^2}}2.
 \label{eq:refined-cubic-maximizer}
\end{align}
It increases from $\sqrt{u_d(u_d-2)}$ to $\sqrt{d-1}$. Stationarity and envelope differentiation yield
\begin{align}
 K_d(r)^3=\frac{1+r^2z(r)^2}{\sqrt{1+z(r)^2}},
 \qquad
 \frac{d\log K_d(r)}{d\log r}=\frac{z(r)^2}{1+z(r)^2}.
 \label{eq:refined-cubic-envelope}
\end{align}
We now maximize $\ell_{d,r}(c)$ over $r$. For $r>r_{0,d}$, its derivative with respect to $\log r$ is $3[(1-z(r)^2)/(1+z(r)^2)-c]$. Since $z(r)$ increases with $r$, the derivative decreases, and the maximum is determined by its sign at the endpoints or by its unique zero. The resulting three branches are
$r=1$ below $c=2/d-1$, the root in
Eq.~\eqref{eq:refined-cubic-root} in the intermediate regime, and $r=r_{0,d}$
above $c=c_{*,d}$. Substitution proves the asserted formulas and their
matching endpoint values. For $r<r_{0,d}$, $K_d(r)=1$, and increasing
$r$ to $r_{0,d}$ cannot reduce the objective.
\end{proof}

The affine branch comes from the contractive endpoint $r=r_{0,d}$. In the intermediate regime, optimizing over $r>r_{0,d}$ while retaining the full norm $K_d(r)$ yields the stronger nonlinear branch.

\subsection{Isotropic states: comparison and tightness}
\label{sec:refined-tightness}

We now specialize the bound to isotropic states, whose normalized entangled fraction is
\begin{align}
\FME(\rho_{p,d})&=1-2\eta_dp,\qquad \eta_d:=1-d^{-2}.
\label{eq:main-isotropic-overlap}
\end{align}
Combining Theorem~\ref{thm:refined-cubic} with independent preparation of each copy gives the cost interval
\begin{align}
B_d(1-2\eta_dp)
&\le\Ec(\rho_{p,d})\le U_d(p):=\Ef(\rho_{p,d}),
\label{eq:main-cubic-cost-sandwich}
\end{align}
for $d\geq3$ and $0\leq p\leq1$. Its width bounds the possible advantage of collective preparation over independent preparation. We compare this interval with the PPT-relative-entropy lower bound below.

\paragraph{Entanglement of formation.}
For isotropic states, the entanglement of formation is known explicitly~\cite{TerhalVollbrecht2000}:
\begin{align}
U_d(p)
&=
\begin{cases}
\displaystyle
\log d-\kappa_dp,
&0\leq p\leq p_{\mathrm{L},d},\\[2pt]
h_2(\gamma_d(p))
+[1-\gamma_d(p)]\log(d-1),
&p_{\mathrm{L},d}<p<p_{\mathrm{H},d},\\[2pt]
0,
&p_{\mathrm{H},d}\leq p\leq1.
\end{cases}
\label{eq:main-isotropic-eof}
\end{align}
The parameters and auxiliary function are
\begin{align}
\kappa_d
&:=\frac{d^2-1}{d(d-2)}\log(d-1),
\label{eq:main-qudit-kappa}\\
\gamma_d(p)
&:=\frac1d\left(\sqrt{1-\eta_d p}+\sqrt{(d-1)\eta_d p}\right)^2,
\label{eq:main-isotropic-gamma}\\
p_{\mathrm{L},d}
&:=\frac{(d-2)^2}{d^2-1},\qquad
p_{\mathrm{H},d}:=\frac{d}{d+1}.
\label{eq:main-isotropic-noise-boundaries}
\end{align}
The state is separable for $p\geq p_{\mathrm H,d}$, while $p_{\mathrm L,d}$ separates the affine and nonlinear branches of $U_d$.

\paragraph{PPT relative entropy of entanglement.}
As a benchmark, we use the PPT relative entropy of entanglement, the best previously known lower bound on $\Ec(\rho_{p,d})$. For isotropic states, its value is~\cite[Theorem~7]{Rains1999}
\begin{align}
R_d(p)&:=
\begin{cases}
\log d-h_2(\eta_dp)-\eta_dp\log(d-1),&p<p_{\mathrm H,d},\\
0,&p\ge p_{\mathrm H,d}.
\end{cases}
\label{eq:main-isotropic-ppt-lower}
\end{align}
Figure~\ref{fig:main-qudit-depolarizing-bounds} compares the cubic-norm and PPT-relative-entropy lower bounds with the entanglement of formation at $d=3,10,60$. In these examples, the cubic-norm bound substantially improves the PPT benchmark and closely follows the upper bound.

\begin{figure}[!ht]
\centering
\includegraphics[width=\textwidth]{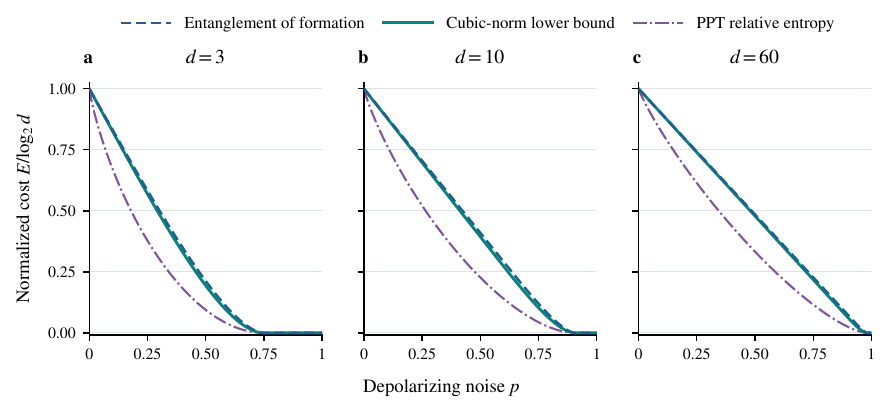}
\caption{Entanglement-cost bounds for isotropic states at $d=3,10,60$, normalized by $\log d$. The solid teal curve is the cubic-norm lower bound $B_d(1-2\eta_dp)$, the purple dash-dotted curve is the PPT-relative-entropy lower bound $R_d(p)$, and the blue dashed curve is the entanglement of formation $U_d(p)$.}
\label{fig:main-qudit-depolarizing-bounds}
\end{figure}

\paragraph{Remaining gap.}
To quantify the uncertainty in the entanglement cost uniformly over the noise parameter, define the maximum interval width
\begin{align}
G_d&:=\max_{0\leq p\leq1}\bigl[U_d(p)-B_d(1-2\eta_dp)\bigr].
\label{eq:main-cubic-maximum-gap}
\end{align}
Both $U_d(p)$ and $B_d(1-2\eta_dp)$ vanish for $p\geq p_{\mathrm H,d}$, so only the entangled regime contributes to $G_d$.
Figure~\ref{fig:main-qudit-max-gap-dimension} plots the normalized interval width $G_d/\log d$ for every integer $3\leq d\leq10^4$. The largest value in this scan is approximately $1.989\%$, attained at $d=8$. Thus the gap between the cubic-norm lower bound and the entanglement of formation remains below $2\%$ of $\log d$ throughout the scanned range. Appendix~\ref{sec:refined-asymptotic} further proves that $G_d\to0$ as $d\to\infty$, so the absolute gap vanishes uniformly over the noise parameter.

\begin{figure}[H]
\centering
\includegraphics[width=\textwidth]{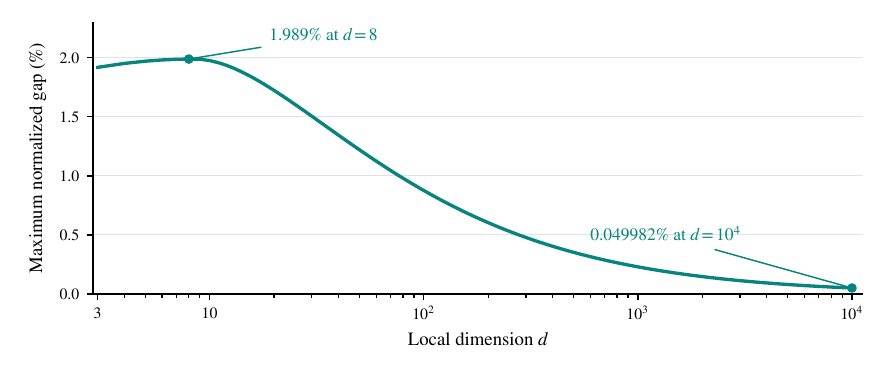}
\caption{Maximum interval width $G_d$, defined in Eq.~\eqref{eq:main-cubic-maximum-gap}, normalized by $\log d$ and expressed as a percentage. The scan includes every integer $3\leq d\leq10^4$, shown on a logarithmic dimension axis. At $d=10^4$, the normalized width is approximately $0.049982\%$.}
\label{fig:main-qudit-max-gap-dimension}
\end{figure}
\FloatBarrier

\section{Conclusion}
\label{sec:main-conclusion}

We have determined the entanglement cost of every qubit isotropic state and, more generally, every two-qubit Bell-diagonal state. For qudit isotropic states, the cubic-norm bound substantially improves the PPT-relative-entropy benchmark and nearly coincides the entanglement of formation. The remaining interval vanishes in absolute width with increasing dimension, uniformly over the noise parameter. The exact entanglement cost for qudit isotropic state remains open.

These results also characterize the entanglement resources needed to simulate depolarizing channels. Channel entanglement cost is the minimum asymptotic Bell-pair rate required for simulation by local operations and classical communication. For depolarizing channels, both parallel and adaptive sequential simulation costs equal the entanglement cost of the corresponding isotropic Choi state~\cite[Theorem~1 and Section~IV]{Wilde2018}. Our results therefore determine the simulation cost exactly for qubit depolarizing channels and give the same narrow cost intervals for qudit channels.

\bigskip
\paragraph{Acknowledgements.}
K.F. is supported in part by the National Natural Science Foundation of China (Grant Nos. 92470113 and 12404569), the Shenzhen Science and Technology Program (Grant Nos. QNXMB20250701091826036 and JCYJ20240813113519025), the Shenzhen Fundamental Research Program (Grant No. JCYJ20241202124023031), the General R\&D Projects of 1+1+1 CUHK-CUHK(SZ)-GDST Joint Collaboration Fund (Grant No. GRDP2025-022), the Guangdong Provincial Quantum Science Strategic Initiative (Grant No. GDZX2503001), and the University Development Fund (Grant No. UDF01003565). 
We acknowledge that OpenAI GPT-6 Astra was used to explore proof strategies, plot figures and support the writing of the manuscript. We
have verified AI-assisted material to the best of our knowledge
and take full responsibility for the content of this work.

\begingroup
\normalsize
\bibliographystyle{alpha_abbrv}
\bibliography{tensorized_bell_fidelity_qubit_cost}

\newcommand{\etalchar}[1]{$^{#1}$}
\begin{thebibliography}{VADDM01}

\bibitem[BDR20]{BeigiDattaRouze2020}
Salman Beigi, Nilanjana Datta, and Cambyse Rouz{\'e}.
\newblock Quantum reverse hypercontractivity: Its tensorization and application to strong converses.
\newblock {\em Communications in Mathematical Physics}, 376(2):753--794, 2020.

\bibitem[Bei13]{Beigi2013}
Salman Beigi.
\newblock Sandwiched {R{\'e}nyi} divergence satisfies data processing inequality.
\newblock {\em Journal of Mathematical Physics}, 54(12):122202, 2013.

\bibitem[DGO{\etalchar{+}}26]{DongEtAl2026}
Yangjing Dong, Li~Gao, Fengning Ou, Penghui Yao, and Haigang Zhou.
\newblock Dimension-free approximate tensorization of quantum hypercontractivity for qudit depolarizing semigroups.
\newblock {\em arXiv preprint arXiv:2606.17729}, 2026.
\newblock arXiv:2606.17729v3.

\bibitem[HHT01]{HaydenHorodeckiTerhal2001}
Patrick~M. Hayden, Micha\l{} Horodecki, and Barbara~M. Terhal.
\newblock The asymptotic entanglement cost of preparing a quantum state.
\newblock {\em Journal of Physics A: Mathematical and General}, 34:6891--6898, 2001.

\bibitem[Kin14]{King2014}
Christopher King.
\newblock Hypercontractivity for semigroups of unital qubit channels.
\newblock {\em Communications in Mathematical Physics}, 328(1):285--301, 2014.

\bibitem[Rai99]{Rains1999}
Eric~M. Rains.
\newblock Bound on distillable entanglement.
\newblock {\em Physical Review A}, 60:179, 1999.
\newblock Erratum: Physical Review A 63, 019902 (2000).

\bibitem[TV00]{TerhalVollbrecht2000}
Barbara~M. Terhal and Karl Gerd~H. Vollbrecht.
\newblock Entanglement of formation for isotropic states.
\newblock {\em Physical Review Letters}, 85:2625--2628, 2000.

\bibitem[VADDM01]{VerstraeteEtAl2001}
Frank Verstraete, Koenraad Audenaert, Jeroen Dehaene, and Bart De~Moor.
\newblock A comparison of the entanglement measures negativity and concurrence.
\newblock {\em Journal of Physics A: Mathematical and General}, 34:10327--10332, 2001.

\bibitem[Wil18]{Wilde2018}
Mark~M. Wilde.
\newblock Entanglement cost and quantum channel simulation.
\newblock {\em Physical Review A}, 98:042338, 2018.

\bibitem[Woo98]{Wootters1998}
William~K. Wootters.
\newblock Entanglement of formation of an arbitrary state of two qubits.
\newblock {\em Physical Review Letters}, 80:2245--2248, 1998.

\end{thebibliography}
\endgroup

% \clearpage
\appendix
\section{Calculations for the tuned qubit semigroup}
\label{app:tuned-qubit-calculations}

We supply the calculations for the qubit semigroup constructed in the proof of Lemma~\ref{lem:main-supporting-line-entropy}. We first derive the matching identity, then verify the semigroup hypotheses and sharp log-Sobolev constant. Throughout, we use the parameters and operators defined in that proof.

\subsection{Dirichlet form and algebraic matching}
\label{app:tuned-qubit-matching}

For an arbitrary complex $2\times2$ matrix $M$, put $X=\sigma_c^{-1/4}M\sigma_c^{-1/4}$. Substituting the generator~\eqref{eq:main-tuned-generator} into the weighted inner product gives
\begin{align}
\langle X,\cL_c(X)\rangle_{\sigma_c}
&=\left|\sqrt{\nu_c}M_{00}-\sqrt{\mu_c}M_{11}\right|^2
+\eta_c\left(|M_{01}|^2+|M_{10}|^2\right).
\label{eq:main-tuned-general-matrix-energy}
\end{align}
The reference-state term is equally explicit:
\begin{align}
-\frac12\tr\!\left[(M^\dagger M+MM^\dagger)\ln\sigma_c\right]
&=-\ln\mu_c\,|M_{00}|^2-\ln\nu_c\,|M_{11}|^2
\nonumber\\
&\quad-\frac12\ln(\mu_c\nu_c)\left(|M_{01}|^2+|M_{10}|^2\right).
\label{eq:main-tuned-general-reference-term}
\end{align}

For $n$ pairs and a fixed site $i$, group the remaining row and column indices into $\boldsymbol\alpha$ and $\boldsymbol\beta$, and write $M$ as a $2\times2$ block matrix with blocks $M_{ab}^{(i)}=(M_{a\boldsymbol\alpha,b\boldsymbol\beta})_{\boldsymbol\alpha,\boldsymbol\beta}$. Taking the partial trace of $\psi$ gives
\begin{align}
p_{ab}^{(i)}&=\lVert M_{ab}^{(i)}\rVert_2^2,
&z_i&=\tr\!\left[M_{00}^{(i)}(M_{11}^{(i)})^\dagger\right].
\label{eq:main-coefficient-block-marginals}
\end{align}
Let $\tau_c:=\sigma_c^{\otimes(n-1)}$, $s_0:=\mu_c$, and $s_1:=\nu_c$. After putting site $i$ first, the blocks of $X$ are
\begin{align}
X_{ab}^{(i)}&=(s_as_b)^{-1/4}\tau_c^{-1/4}M_{ab}^{(i)}\tau_c^{-1/4}.
\end{align}
The product weights cancel on the other sites, so the blocks of $\Gamma_{\Sigma_c}^{1/2}(\cL_{c,i}(X))$ are
\begin{align}
\begin{pmatrix}
\nu_c M_{00}^{(i)}-\sqrt{\mu_c\nu_c}M_{11}^{(i)} & \eta_c M_{01}^{(i)}\\
\eta_c M_{10}^{(i)} & \mu_c M_{11}^{(i)}-\sqrt{\mu_c\nu_c}M_{00}^{(i)}
\end{pmatrix}.
\end{align}
Using $\langle X,\cL_{c,i}(X)\rangle_{\Sigma_c}=\tr[M^\dagger\Gamma_{\Sigma_c}^{1/2}(\cL_{c,i}(X))]$ therefore gives
\begin{align}
\langle X,\cL_{c,i}(X)\rangle_{\Sigma_c}
&=\left\lVert\sqrt{\nu_c}M_{00}^{(i)}-\sqrt{\mu_c}M_{11}^{(i)}\right\rVert_2^2
+\eta_c\left(\lVert M_{01}^{(i)}\rVert_2^2+\lVert M_{10}^{(i)}\rVert_2^2\right).
\label{eq:main-tuned-block-energy}
\end{align}
Summing over $i$ proves Eq.~\eqref{eq:main-common-tangent-energy}. Likewise,
\begin{align}
\ln\Sigma_c
&=\sum_{i=1}^n\one_2^{\otimes(i-1)}\otimes(\ln\sigma_c)\otimes\one_2^{\otimes(n-i)}.
\label{eq:main-product-log-reference}
\end{align}
Each summand acts on the $i$th qubit and as the identity on the other qubits in the same register. Applying the one-site reference formula to the blocks proves Eq.~\eqref{eq:main-common-tangent-reference}. No factorization of $M$ is used.

It remains to match coefficients. We choose the rates so that the reference-state term minus the Dirichlet form divided by $a_c$ depends only on $\lVert M\rVert_2^2$ and $|M_{00}+M_{11}|^2$. The coefficients of $|M_{00}|^2$ and $|M_{11}|^2$ must agree, and their difference from the coefficient of $|M_{01}|^2$ must equal half the coefficient of $\operatorname{Re}(\overline{M_{00}}M_{11})$. These conditions are
\begin{align}
\ln\frac{\nu_c}{\mu_c}
&=\frac{\nu_c-\mu_c}{a_c},
&\frac{\eta_c-1/2}{a_c}
&=\frac{\sqrt{\mu_c\nu_c}}{a_c}.
\label{eq:main-tuned-matching-conditions}
\end{align}
For $0<c<1$, they give precisely Eq.~\eqref{eq:main-tuned-rates}. Substitution then gives
\begin{align}
-\ln\mu_c-\frac{\nu_c}{a_c}
&=-\ln\nu_c-\frac{\mu_c}{a_c}
=(\ln2)\left[\W(c)-c\W'(c)\right],
\nonumber\\
-\frac12\ln(\mu_c\nu_c)-\frac{\eta_c}{a_c}
&=(\ln2)\left[\W(c)-(c+1)\W'(c)\right],\\
\frac{\sqrt{\mu_c\nu_c}}{a_c}
&=(\ln2)\W'(c).
\label{eq:main-tuned-entry-identities}
\end{align}
These identities give the supporting-line expression in Eq.~\eqref{eq:main-common-tangent-entropy}. They remain valid at $c=1$ with $a_1=1/2$ and $\eta_1=1$.

The slope is fixed by $\sqrt{\mu_c\nu_c}/a_c=(\ln2)\W'(c)$, and the contact state determines the intercept. Indeed, at $M=\sqrt{\sigma_c}$ the Dirichlet form vanishes, while the reference-state term equals $-\tr(\sigma_c\ln\sigma_c)=(\ln2)\W(c)$.

\subsection{Semigroup properties and the sharp log-Sobolev constant}
\label{app:tuned-qubit-lsi}

With the positive-generator convention, $\cL_c$ has the Lindblad representation
\begin{align}
\cL_c(X)
&=
\sum_{k=1}^3
\left[
\frac12\{V_k^\dagger V_k,X\}
-V_k^\dagger XV_k
\right],
\nonumber\\
V_1&:=\sqrt{\nu_c}\ketbra{1}{0},
\qquad
V_2:=\sqrt{\mu_c}\ketbra{0}{1},
\qquad
V_3:=\frac{\sqrt c}{2}Z.
\label{eq:main-tuned-lindblad}
\end{align}
Hence $\Psi_{c,t}=e^{-t\cL_c}$ is completely positive and unital. Direct inspection shows that $\cL_c$ is $\sigma_c$-reversible, preserves the diagonal algebra, and acts on both off-diagonal matrix units with eigenvalue $\eta_c$. Thus $\sigma_c$ is invariant, and each map $\Psi_{c,t}$ is PODS. Since $\mu_c,\nu_c,\eta_c>0$, the diagonal dynamics has a unique invariant state and all coherences decay, proving primitivity.

We now prove $\alpha_2(\cL_c)=a_c$. Write $\delta_c:=\nu_c-\mu_c=\sqrt{1-c^2}$. Let $X>0$, normalize $\lVert X\rVert_{\sigma_c,2}=1$, and put
\begin{align}
M
&:=\sigma_c^{1/4}X\sigma_c^{1/4}
=\frac12\left(s\one+t\boldsymbol n\cdot\boldsymbol\sigma\right),
\label{eq:main-tuned-bloch}
\end{align}
where $s>t\geq0$, $\lVert\boldsymbol n\rVert=1$, and $s^2+t^2=2$. Write $u:=n_z$. The normalization gives $\tr M^2=1$, and therefore $\operatorname{Ent}_{2,\sigma_c}(X)=D(M^2\Vert\sigma_c)$. A direct calculation then gives
\begin{align}
\cE_{2,\cL_c}(X)
&=
\frac{(1-c)s^2+(1+c)t^2}{4}
+\frac{\delta_cst}{2}u,
\label{eq:main-tuned-energy-orientation}\\
D(M^2\Vert\sigma_c)
&=
D_{0,c}(s,t)
+\frac{st}{2}u\ln\frac{\nu_c}{\mu_c},
\label{eq:main-tuned-entropy-orientation}\\
D_{0,c}(s,t)
&:=
\sum_{\pm}
\frac{(s\pm t)^2}{4}
\ln\frac{(s\pm t)^2}{4}
-\frac12\ln(\mu_c\nu_c).
\label{eq:main-tuned-entropy-isotropic-part}
\end{align}
Combining these identities gives
\begin{align}
\cE_{2,\cL_c}(X)
-a_cD(M^2\Vert\sigma_c)
&=
\frac{(1-c)s^2+(1+c)t^2}{4}
-a_cD_{0,c}(s,t).
\label{eq:main-tuned-orientation-independent}
\end{align}
The terms proportional to $u$ cancel because
\begin{align}
a_c\ln\frac{\nu_c}{\mu_c}
&=\delta_c.
\label{eq:main-tuned-cancellation}
\end{align}
Thus the difference in Eq.~\eqref{eq:main-tuned-orientation-independent} depends only on the eigenvalues of $M$. We may evaluate it at $u=1$, where $M$ and $X$ are diagonal. On this algebra, $\cL_c$ is the classical two-point generator with invariant probabilities $(\mu_c,\nu_c)$ and sharp $2$-log-Sobolev constant
\begin{align}
\frac{\nu_c-\mu_c}{\ln\nu_c-\ln\mu_c}
&=a_c.
\label{eq:main-classical-two-point-alpha}
\end{align}
The classical inequality therefore proves $\alpha_2(\cL_c)\geq a_c$. Restricting to diagonal $X$ gives the reverse inequality and hence equality. The endpoint $c=1$ follows by continuity: the displayed expressions extend with $a_1=1/2$, which is also the sharp constant on the diagonal algebra.

\section{Uniform large-dimension estimate}
\label{sec:refined-asymptotic}

This appendix gives an analytical estimate for the maximum interval width
studied numerically in Figure~\ref{fig:main-qudit-max-gap-dimension}.
Although conservative in finite dimensions, it proves that the width
vanishes in absolute ebits, uniformly over the entire noise range.

\begin{corollary}
\label{cor:refined-vanishing-gap}
Let $U_d(p):=\Ef(\rho_{p,d})$ and let $B_d$ be the cubic-norm
lower bound in Eq.~\eqref{eq:refined-cubic-profile}. Write
\begin{align}
G_d&:=\max_{0\leq p\leq1}
\bigl\{U_d(p)-B_d(1-2\eta_dp)\bigr\}.
\end{align}
For every integer $d\geq3$,
\begin{align}
0\leq G_d
&\leq \log d\left[1-\frac{d\log d}{(d-1)s_d}\right],\quad s_d:=3\log\!\left[(d-1)^{1/3}+(d-1)^{-1/3}\right].
\label{eq:refined-uniform-gap-cap}
\end{align}
In particular,
\begin{align}
\limsup_{d\to\infty}d^{2/3}G_d
&\leq\frac3{\ln2},
&
\frac{G_d}{\log d}&=O\!\left(\frac{d^{-2/3}}{\log d}\right).
\label{eq:refined-vanishing-gap-rate}
\end{align}
\end{corollary}

\begin{proof}
Put $A=\log d$ and $a=dA/(d-1)$. Convexity of isotropic entanglement of
formation, together with its endpoint values $U_d(0)=A$ and $U_d(p_{\mathrm H,d})=0$,
gives the chord bound
\begin{align}
U_d(p)&\leq A-a\eta_dp,\qquad 0\leq p\leq p_{\mathrm H,d}.
\end{align}
The global affine bound in Proposition~\ref{prop:refined-cubic-explicit} gives
$B_d(1-2\eta_dp)\geq[A-s_d\eta_dp]_+$.
Validity of this lower bound at the separable endpoint implies $s_d\geq a$;
hence its clipping point $p_0=A/(s_d\eta_d)$ belongs to $[0,p_{\mathrm H,d}]$.
The difference between the chord and this clipped affine function increases
up to $p_0$ and decreases thereafter. Its maximum is therefore
\begin{align}
A-a\eta_dp_0&=A\left(1-\frac a{s_d}\right),
\end{align}
which proves Eq.~\eqref{eq:refined-uniform-gap-cap}. For $p\geq p_{\mathrm H,d}$, both
endpoints of the certified interval are zero.
Finally,
\begin{align}
s_d&=\log(d-1)+3\log\!\left[1+(d-1)^{-2/3}\right],
\\
s_d-a&=\frac3{\ln2}\,d^{-2/3}+o(d^{-2/3}),
\qquad \frac{A}{s_d}\to1.
\end{align}
Substitution gives Eq.~\eqref{eq:refined-vanishing-gap-rate}.
\end{proof}

\end{document}